\documentclass[final,5p,times,twocolumn]{elsarticle}
\usepackage{amsthm,amssymb,amsmath}
\usepackage{amsfonts}
\usepackage[T1]{fontenc}
\usepackage[latin1]{inputenc}
\usepackage[english]{babel}
\usepackage{enumitem}
\usepackage{graphicx}
\usepackage{titlesec}
\usepackage[font=small]{caption}
\journal{Automatica}

\newcommand{\eps}{\varepsilon}

\titlespacing*{\section}{12pt}{12pt plus 4.0pt minus 4.0pt}{12pt plus 4.0pt minus 4.0pt}

\newtheorem{thm}{Theorem}
\newtheorem{defn}[thm]{Definition}

\newtheorem{cor}[thm]{Corollary}
\theoremstyle{remark}

\begin{document}

\begin{frontmatter}
\title{Deciding Detectability for Labeled Petri Nets}

\author[tm]{Tom{\' a}{\v s}~Masopust}\ead{masopust{@}math.cas.cz}
\author[xy,xy2]{Xiang Yin\corref{cor1}}\ead{yinxiang@sjtu.edu.cn}

\cortext[cor1]{This paper was not presented at any conference. Corresponding author X. Yin. Tel. +8613636434613. Fax +862134204522. T.~Masopust was supported by RVO 67985840.}
\address[tm]{Institute of Mathematics, Czech Academy of Sciences, {\v Z}i{\v z}kova 22, 616 62 Brno, Czechia}
\address[xy]{Department of Automation, Shanghai Jiao Tong University, Shanghai 200240, China}
\address[xy2]{Key Laboratory of System Control and Information Processing, Ministry of Education of China, Shanghai 200240, China \vspace{-20pt}}


\begin{abstract}
  Detectability of discrete event systems (DESs) is a property to decide \emph{a priori\/} whether the current and subsequent states can be determined based on observations. We investigate the existence of algorithms for the verification of strong and weak detectability for DESs modeled as labeled Petri nets (LPNs). Strong detectability requires that we can \emph{always} determine, after a finite number of observations, the current and subsequent markings of the system, while weak detectability requires that we can determine, after a finite number of observations, the current and subsequent markings for \emph{some} trajectories of the system. We show that there exists an algorithm to check strong detectability that requires at least exponential space, and that there is no algorithm to check weak detectability. Our results extend the existing studies on the verification of detectability from finite-state automata to LPNs. As a corollary, we obtain that current-state opacity is undecidable even if the secret set is a singleton, which improves the known results.
\end{abstract}

\end{frontmatter}

\section{Introduction}
  State estimation is one of the central problems in systems and control theory. It plays a key role in many problems where one needs to estimate the state of the system based on imperfect observations. We investigate an important property of the state estimation problem called {\em detectability\/} for discrete event systems (DESs) modeled by labeled Petri nets (LPNs).

  In the context of DESs, the problem of state estimation has been extensively investigated due to its applications in many different problems~\cite{OzverenW1990,Ramadge1986,ShuLinYing2007}. In particular, Shu and Lin~\cite{ShuLinYing2007} proposed the concept of detectability for DES modeled by finite-state automata that characterizes \emph{a priori\/} whether or not the current and subsequent states can be determined based on observations. The property of detectability has drawn a considerable attention in the literature~\cite{keroglou2017verification,ShuLin2011,shu2013delayed,yin2017initial}, including the complexity studies on the verification of different notions of detectability~\cite{Masopust2017d,YinLafortune17,Zhang17} and the generalization of the notion to, e.g., stochastic discrete-event systems~\cite{keroglou2015detectability,keroglou2017verification,yin2017initial}.

  In this paper, we study the existence of algorithms for the verification of strong and weak detectability in the context of the work of Shu and Lin~\cite{ShuLinYing2007}, generalized from finite-state automata to unbounded LPNs. Specifically, we assume that both the Petri net structure and the initial marking are known, and that the system is partially-observed via a labeling function.

  Strong detectability requires that we can \emph{always} determine, after a finite number of observations, the current and subsequent markings of the system, and weak detectability requires that we can determine, after a finite number of observations, the current and subsequent markings for \emph{some} trajectories of the system.

  For systems modeled by finite-state automata, one can construct an algorithm checking strong detectability in polynomial time~\cite{ShuLin2011} (actually, one can construct an efficient parallel algorithm~\cite{Masopust2017d}). On the other hand, any algorithm checking weak detectability requires at least polynomial space, since this problem is PSPACE-complete~\cite{YinLafortune17,Zhang17}. It is open whether there are polynomial-time algorithms for PSPACE-complete problems, but it is known that there are no efficient parallel algorithms for PSPACE-complete problems. The results for weak detectability hold even for a very restricted type of automata~\cite{Masopust2017d}.

  For systems modeled by bounded LPNs, the results for automata imply that both of these properties are decidable for bounded LPNs, since we can explicitly enumerate the reachable markings and use the verification techniques for automata.

  However, whether the properties are also decidable for unbounded LPNs is no longer straightforward, because the reachable space of such a system is infinite in general. Very recently, Zhang and Giua~\cite{ZhangGiua2018} showed undecidability of weak detectability for LPNs with inhibitor arcs, which are computationally universal models, and stated the decidability questions of strong and weak detectability for LPNs, which are not computationally universal, as open problems. We resolve these questions.

  First, we show that verifying strong detectability for LPNs is decidable by expressing the property as a path formula in \emph{Yen's logic}, for which satisfiability was shown decidable by reduction to reachability~\cite{yen1992unified,AtigH11}. Hence strong detectability is reducible to reachability as well. We further show that deciding strong detectability is EXPSPACE-hard, and hence any algorithm verifying strong detectability requires at least exponential space, and is thus infeasible. If the conjecture that reachability is in EXPSPACE is true, then deciding strong detectability is EXPSPACE-complete.

  Then, we show that checking weak detectability for LPNs is undecidable, solving thus the second open problem that improves the recent result of Zhang and Giua~\cite{ZhangGiua2018}. We prove the result by reducing the \emph{language inclusion problem\/} of two LPNs to the weak detectability verification problem. Our proof is similar, but more involved, than the construction of Tong et al.~\cite{TongLSG17} showing that the current-state opacity problem is undecidable. The secret set in the construction of Tong et al.\ is as large as the reachable set of one of the Petri nets under consideration\footnote{Tong et al. write the secret set as $S=\{\bar{p_3}\}$, which is the set of all markings with a token in place $p_3$.}, and hence infinite in general. It is a natural question whether undecidability of current-state opacity follows from the infinity of the secret set. In other words, whether current-state opacity is decidable if the secret set is finite. As a consequence of our result, we show that current-state opacity is undecidable even if the secret set consists of a single marking. This result strengthens and completes the study of Tong et al.~\cite{TongLSG17}.

  Our work is related to several works on state estimation of Petri nets~\cite{basile2015state,GiuaSeatzu2002,ramirez2003observability,ru2010sensor,tong2016equivalence,ZhangGiua2018}. In particular, it is closely related to the work of Giua and Seatzu~\cite{GiuaSeatzu2002} who proposed several different observability properties for (unlabeled) place/transition nets. Specifically, they proposed two observability properties---marking observability and strong marking observability; the former requires that there exists a word under which the marking of the system can be precisely determined, while  the latter requires that the marking of the system can be precisely determined after a finite delay $k$.

  Marking observability and strong observability are similar (but not identical) notions to weak and strong detectability, respectively. The main difference between our results and the results of Giua and Seatzu is that LPNs are more general than the unlabeled models used by Giua and Seatzu, which is also reflected in the results---we show that weak detectability for LPNs is undecidable whereas Giua and Seatzu show that marking observability for their unlabeled models is decidable. Moreover, in the case of strong marking observability, there is a given pre-specified detection bound $k$. Therefore, this property is trivially decidable by explicitly enumerating the reachable markings of the system within $k$ steps. Notice that we do not pre-specify any such detection bound for checking strong detectability, which makes the verification of strong detectability for unbounded LPNs non-trivial because the search space is infinite in general.

  Our work is also related to the work of Ram{\'\i}rez-Trevi{\~n} et al.~\cite{ramirez2003observability}, who proposed marking detectability, which is a property closely related to strong detectability. However, Ram{\'\i}rez-Trevi{\~n} et al.\ only provide sufficient conditions for checking marking detectability and, to the best of our knowledge, (un)decidability of checking strong and weak detectability in the context of Shu and Lin~\cite{ShuLinYing2007} for LPNs has not been established in the literature so far.

  Finally, we would like to point out that detectability is a property that determines \emph{a priori\/} whether the marking of the system can be detected.  On the other hand, there is a large body of the literature on the \emph{online\/} marking estimation for Petri nets. This topic is, however, beyond the scope of this paper; an interested reader is referred to the literature~\cite{basile2015state,cabasino2017marking,dotoli2009line} for more details.

\section{Preliminaries and Definitions}\label{sec:2}
  We assume that the reader is familiar with the basic notions of Petri nets~\cite{Peterson1981}. For a set $A$, $|A|$ denotes the cardinality of $A$. An {\em alphabet\/} $\Sigma$ is a finite nonempty set (of {\em events}). A {\em word\/} over $\Sigma$ is a sequence of events of $\Sigma$. Let $\Sigma^*$ denote the set of all finite words over $\Sigma$, where the {\em empty word\/} is denoted by $\eps$, and let $\Sigma^{\omega}$ denote the set of all infinite words over $\Sigma$. For a word $u \in \Sigma^*$, $|u|$ denotes its length. Let $\mathbb{N} = \{0,1,2,\ldots\}$ denote the set of all non-negative integers.

  A {\em Petri net\/} is a structure $N=(P,T,Pre,Post)$, where $P$ is a finite set of {\em places}, $T$ is a finite set of {\em transitions}, $P \cup T \neq \emptyset$ and $P \cap T = \emptyset$, and $Pre\colon P \times T \to \mathbb{N}$ and $Post\colon P \times T \to \mathbb{N}$ are the pre- and post-incidence functions specifying the arcs directed from places to transitions and vice versa, respectively.
  A {\em marking\/} is a function $M\colon P \to \mathbb{N}$ that assigns to each place a number of tokens.
  A {\em Petri net system\/} $(N, M_0)$ is the Petri net $N$ with the initial marking $M_0$.
  A transition $t$ is {\em enabled\/} in a marking $M$ if $M(p) \ge Pre(p,t)$ for every place $p\in P$. An enabled transition $t$ can {\em fire\/} and the resulting marking $M'$ is defined as $M'(p) = M(p) - Pre(p,t) + Post(p,t)$ for every $p\in P$. We write $M\xrightarrow{\sigma}_{N}$ to denote that the sequence of transitions $\sigma$ is enabled in the marking $M$ of $N$, and $M\xrightarrow{\sigma}_N M'$ to denote that the firing of the sequence of transitions $\sigma$ results in a marking $M'$. For simplicity, we omit the subscript $N$ if the net is clear from the context. We write $L(N,M_0) = \{ \sigma \in T^* \mid M_0\xrightarrow{\sigma} \}$ to denote the set of all transition sequences enabled in the marking $M_0$.
  A marking $M$ is {\em reachable\/} in the Petri net system $(N,M_0)$ if there is a sequence of transitions $\sigma \in T^*$ such that $M_0\xrightarrow{\sigma} M$. The set of all markings reachable from the marking $M_0$ defines the reachability set of the Petri net system $(N,M_0)$, denoted by $R(N,M_0)$.

  A {\em labeled Petri net system\/} is a quadruple $G=(N,M_0,\Sigma,\ell)$, where $(N,M_0)$ is a Petri net system, $\Sigma$ is an alphabet (a set of labels), and $\ell\colon T \to \Sigma\cup\{\eps\}$ is a labeling function that assigns to each transition $t \in T$ a symbol from $\Sigma\cup\{\eps\}$. The labeling function can be extended to $\ell\colon T^* \to \Sigma^*$ defining $\ell(\sigma t) = \ell(\sigma)\ell(t)$ for $\sigma \in T^*$ and $t \in T$; we define $\ell(\lambda) = \eps$ for the empty transition sequence $\lambda$.
  We say that a transition $t\in T$ is observable if $\ell(t)\in\Sigma$; unobservable otherwise. The {\em language\/} of $G$ is defined as the set $L(G) = \{ \ell(\sigma) \mid \sigma \in L(N,M_0) \}$. Similarly, $L^\omega(G)$ denotes the set of all infinite words generated by $G$.
  Finally, for a word $s\in L(G)$, $R(G,s) = \{ M \mid \sigma\in L(N,M_0),\, \ell(\sigma)=s,\, M_0\xrightarrow{\sigma}M \}$ denotes the set of all reachable markings consistent with the observation $s$.

  As usual when detectability is discussed~\cite{ShuLin2011}, we make the following two assumptions on the system $G$:
    (i) $G$ is {\em deadlock free}, that is, in every reachable marking of the system, there is at least one transition that can fire, and
    (ii) $G$ cannot generate an infinite unobservable sequence. Notice that for finite-state systems, this assumption is equivalent to avoiding cycles of unobservable transitions.

  Considering the checking of these assumptions. Deadlock-freedom is reducible to reachability, and hence it is decidable, and EXPSPACE-hard. The existing algorithms use non-primitive recursive space~\cite{esparza}. Checking the second assumption is EXPSPACE-complete (see~\ref{appA}).

\section{Strong Detectability}\label{sec:3}
  Strong detectability is a property requiring that we can determine, after a finite number of observations, the current and subsequent states for all trajectories of the system. This property is formally defined as follows.

  \begin{defn}
    An LPN system $G=(N,M_0,\Sigma,\ell)$ is {\em strongly detectable\/} if there exists an integer $n \ge 0$ such that for every infinite word $s  \in L^{\omega}(G)$ and every finite prefix $s'$ of $s$, if $s'$ is longer than $n$, then $|R(G,s')|=1$.
  \end{defn}

  To check strong detectability, it suffices to verify whether or not there are two arbitrarily long sequences with the same observation and leading to two different markings. To formalize this idea, we use the \emph{twin-plant\/} construction for Petri nets used in the literature to test diagnosability~\cite{cabasino2012new,yin2017decidability} and prognosability~\cite{yin2018prognosis}.

  Let $G=(N,M_0,\Sigma,\ell)$ be an LPN, and let $G'=(N',M_0',\Sigma,\ell)$ be a place-disjoint copy of $G$, that is, $N'=(P',T,Pre',Post')$ where $P'=\{p' \mid p\in P\}$ is a disjoint copy of $P$ and the functions $Pre'$ and $Post'$ are adjusted in the natural way. The copy $G'$ has the same initial marking as $G$, that is, $M_0'(p')=M_0(p)$ for every $p' \in P'$. We define a Petri net $(N_{\|},M_{0,\|})=((P_{\|},T_{\|},Pre_{\|},Post_{\|}),M_{0,\|})$ that is essentially the (label-based) synchronization of $G$ and $G'$, where
  the set of places is $P_{\|}=P\cup P'$,
  the initial marking $M_{0,\|}=[M_0^{\top} \   M_0'^{\top}]^{\top}$ is the concatenation of the initial markings of $G$ and $G'$,
  the transitions $T_{\|}= (T\cup\{\lambda\})\times (T\cup\{\lambda\})\setminus \{(\lambda,\lambda)\}$ are pairs of transitions of $G$ and $G'$ without the empty pair, and the functions
  $Pre_{\|}\colon P_{\|}\times T_{\|}\to \mathbb{N}$ and $Post_{\|}\colon P_{\|}\times T_{\|}\to \mathbb{N}$ are defined as follows:
  \begin{itemize}
    \itemsep0pt
    \item
      for every $p\in P$ and every $t\in T$ with $\ell(t)=\eps$,
      we define $Pre_{\|}( p,(t,\lambda) )= Pre(p,t)$ and $Post_{\|}( p,(t,\lambda) )= Post(p,t)$;
    \item
      for every $p'\in P'$ and every $t\in T$ with $\ell(t)=\eps$,
      we define $Pre_{\|}( p',(\lambda,t) )= Pre'(p',t)$ and $Post_{\|}( p',(\lambda,t) )= Post'(p',t)$;
    \item
      for every $p\in P$ and every $t_1,t_2\in T$ with $\ell(t_1)=\ell(t_2)\neq\eps$,
      we define $Pre_{\|}( p, (t_1,t_2) )=Pre(p,t_1)$ and $Post_{\|}( p, (t_1,t_2) )=Post(p,t_1)$;
    \item
      for every $p'\in P'$ and every $t_1,t_2\in T$ with $\ell(t_1)=\ell(t_2)\neq\eps$,
      we define $Pre_{\|}( p', (t_1,t_2) )=Pre'(p',t_2)$ and $Post_{\|}( p', (t_1,t_2) ) = Post'(p',t_2)$;
    \item
      otherwise, no arc is defined
      ($Pre_{\|}(p,t) = Post_{\|}(p,t)=0$).
  \end{itemize}

  Essentially, $(N_{\|},M_{0,\|})$ is constructed to track all pairs of sequences that have the same observation. More specifically, for any $(\sigma,\sigma')\in L(N_{\|},M_{0,\|})$, we have $\ell(\sigma)=\ell(\sigma')$. On the other hand, for any $\sigma,\sigma'\in L(N,M_0)$ such that $\ell(\sigma)=\ell(\sigma')$, there exists a sequence in $(N_{\|},M_{0,\|})$ whose first and second components are $\sigma$ and $\sigma'$, respectively (possibly by inserting the empty transition sequence $\lambda$). For an example illustrating the construction, we refer the reader to the literature~\cite{cabasino2012new,yin2018prognosis}.

  The following result shows how to use the structure $(N_{\|},M_{0,\|})$ to verify strong detectability.
  \begin{thm}\label{thm_formula}
    An LPN $G=(N,M_0,\Sigma,\ell)$ is not strongly detectable if and only if, in $(N_{\|},M_{0,\|})$, there exists a sequence
    \[
      M_{0,\|}
      \xrightarrow{\ \alpha\ }_{N_{\|}}  M_1
      \xrightarrow{\ \beta\ }_{N_{\|}}   M_2
      \xrightarrow{\ \gamma\ }_{N_{\|}}   M_3
    \]
    such that
    $
      (M_1\le M_2)
      \land
      |\beta|>0
      \land
      \bigvee_{p\in P} M_3(p)\neq M_3(p').
    $
  \end{thm}
  \begin{proof}
    ($\Leftarrow$)
    Suppose that there is such a sequence. Let $M_{i,1}$ and $M_{i,2}$, for $i=1,2,3$, denote the first and the second components of $M_i$, respectively, that is, $M_i=[M_{i,1}^{\top}\ M_{i,2}^{\top}]^{\top}$ where the lengths of $M_{i,1}$ and $M_{i,2}$ coincide and are equal to the number of places in $G$. Let $\alpha=(\alpha_1,\alpha_2)$, $\beta=(\beta_1,\beta_2)$, and $\gamma=(\gamma_1,\gamma_2)$. By the construction of $N_{\|}$, $\ell(\alpha_1)=\ell(\alpha_2)$, $\ell(\beta_1)=\ell(\beta_2)$, and $\ell(\gamma_1)=\ell(\gamma_2)$. Since $|\beta|>0$, either $\beta_1$ or $\beta_2$ is not the empty transition; without loss of generality, let $\beta_1\not=\lambda$.

    Let $n\in \mathbb{N}$ be an arbitrary natural number. We consider an infinite sequence
    $
      \alpha_1\beta_1^{m+1}\gamma_1 w \in L^{\omega}(G),
    $
    where $w$ is an arbitrary infinite continuation of the sequence $\sigma_1=\alpha_1\beta_1^{m+1}\gamma_1$ such that $\ell(w)\neq\eps$; such a continuation exists by the assumptions that the system is deadlock free and there is no infinite unobservable sequence. The sequence $\sigma_1$ is well defined in $G$ because $M_1\leq M_2$, and hence the sequence $\sigma_2=\alpha_2\beta_2^{m+1}\gamma_2 \in L(G)$ is also well defined in $G$.
    Let $M_0\xrightarrow{\sigma_1}_N M_{\sigma_1}$ and  $M_0\xrightarrow{\sigma_2}_N M_{\sigma_2}$. Then
    \[
      M_{\sigma_i}= M_{i,3} + m \cdot (M_{i,2}-M_{i,1})\,.
    \]
    Let $p$ be a place such that $M_3(p)\not= M_3(p')$. Then we can always find an integer $m\geq n$ such that $M_{\sigma_1}(p)\not=M_{\sigma_2}(p')$. Since $s = \ell(\alpha_1\beta_1^{m+1}\gamma_1) = \ell(\alpha_2\beta_2^{m+1}\gamma_2)$ is a prefix of $\ell(\alpha_1\beta_1^{m+1}\gamma_1 w)$, we have that $\{M_{\sigma_1},M_{\sigma_2}\}\subseteq R(G,s)$, and hence $|R(G,s)|>1$. Moreover, $M_1\le M_2$ implies the existence of $\beta_1^{\omega}$ in $G$, and hence $\ell(\beta_1)\neq\eps$, because $\ell(\beta_1)=\eps$ would give $\ell(\beta_1^{\omega})=\eps$, which contradicts the assumption that no such sequence exists. Therefore, $|s|\geq m+1>n$. Since $n$ was chosen arbitrarily, the system is not strongly detectable.

    ($\Rightarrow$)
    Suppose that the system is not strongly detectable, that is,
      for every $n\in \mathbb{N}$
      there exist $s\in L^\omega(G)$ and
      a finite prefix $s'$ of $s$ such that
      $|s'|\geq n$ and $|R(G,s')|>1$.
    Then, for any $n\in \mathbb{N}$, there are sequences $\alpha,\beta \in L(N,M_0)$ such that
      (i) $\ell(\alpha) = \ell(\beta)$ and $|\ell(\alpha)| = |\ell(\beta)| \geq n$, and
      (ii) $M_0\xrightarrow{\alpha}_N  M_{\alpha}$ and $M_0 \xrightarrow{\beta}_N  M_{\beta}$ with $M_{\alpha}\not= M_{\beta}$.
    By (i) and the construction of $N_{\|}$,
    there exists a sequence $\sigma\in L(N_{\|},M_{0,\|})$  in $N_{\|}$ such that $\sigma$ is in the form of $\sigma=(\alpha,\beta)$.
    Let $\sigma=t_1t_2\cdots t_k$ for some $t_i\in T_{\|}$ and $k\ge n$, and let $M_1,M_2,\dots,M_{k}$ be the markings induced by the transitions, i.e.,
    $
                           M_{0,\|}
      \xrightarrow{t_1}_{N_{\|}} M_1
      \xrightarrow{t_2}_{N_{\|}} M_2
      \xrightarrow{t_3}_{N_{\|}} \cdots
      \xrightarrow{t_k}_{N_{\|}} M_{k},
    $
    where $M_k=[M_{\alpha}^\top  \ M_{\beta}^\top]^\top$.

  Consider a computation tree consisting of the computations described above. There is such a computation of length at least $n$ for every $n\in\mathbb{N}$, and hence the tree is infinite. Therefore, by K\"onig's lemma~\cite{koenig} stating that every finitely branching infinite tree contains an infinite path, there is an infinite path $C_0,C_1,C_2,\ldots$ in the tree, where $C_0$ is the initial marking $M_{0,\|}$. Then, since vectors of natural numbers with the product order form a well-quasi-ordering, Dickson's lemma~\cite{dickson1913finiteness} implies that there are $i < j$ such that $C_i \le C_j$. Since the tree consists only of computations of the above form, $C_0,C_1,\ldots,C_j$ is a prefix of such a computation, and hence there is a sequence $C_{j+1},\ldots,C_m$ such that $C_0,C_1,\ldots,C_j,C_{j+1},\ldots,C_m$ is a computation of the above form, that is, $C_m$ is of the form $[M_{\alpha}^\top  \ M_{\beta}^\top]^\top$ for some $\alpha$ and $\beta$ satisfying (i) and (ii) above.
    Consider the sequence
    \begin{equation*}
                                                M_{0,\|}
      \xrightarrow{t_1 \cdots t_i}_{N_{\|}}      C_{i}
      \xrightarrow{t_{i+1}\cdots t_{j}}_{N_{\|}} C_{j}
      \xrightarrow{t_{j+1}\cdots t_m}_{N_{\|}}   C_{m}\,.
    \end{equation*}
    Since $C_m = [M_{\alpha}^\top  \ M_{\beta}^\top]^\top$ and $M_{\alpha}\not= M_{\beta}$, there is a place $p$ such that $C_m(p) = M_{\alpha}(p) \neq M_{\beta}(p') = C_m(p')$. Finally, $|t_{i+1}\cdots t_j|>0$, because $i<j$, and hence the sequence satisfies the statement of the theorem.
  \end{proof}

  To state our first result, we briefly recall a fragment of Yen's path logic, the satisfiability of which is decidable~\cite{yen1992unified,AtigH11}. Let $M_1,M_2,\ldots$ be variables representing markings and $\sigma_1,\sigma_2,\ldots$ be variables representing finite sequences of transitions. Every mapping $c \in \mathbb{N}^{|P|}$ is a term. For all $j > i$, if $M_i$ and $M_j$ are marking variables, then $M_j - M_i$ is a term, and if $T_1$ and $T_2$ are terms, then $T_1+T_2$ and $T_1-T_2$ are terms.
  If $c \in \mathbb{N}$ and $t \in T$, then $\#_t(\sigma_1) \le c$ and $\#_t(\sigma_i) \ge c$ are transition predicates, where $\#_t(\sigma)$ denotes the number of occurrences of $t$ in $\sigma$. If $T_1$ and $T_2$ are terms and $p_1,p_2 \in P$ are places, then $T_1(p_1) = T_2(p_2)$, $T_1(p_1) < T_2(p_2)$, and $T_1(p_1) > T_2(p_2)$ are marking predicates. A predicate is a positive boolean combination of transition and marking predicates. A {\em path formula\/} is a formula of the form
  $
    (\exists \sigma_1, \sigma_2,\ldots, \sigma_n)
    (\exists M_1,\ldots, M_n)
    (M_0 \xrightarrow{\sigma_1} M_1 \xrightarrow{\sigma_2} \cdots \xrightarrow{\sigma_n} M_n)
    \land
    \varphi(M_1,\ldots,M_n,\sigma_1,\ldots,\sigma_n)
  $
  where $\varphi$ is a predicate.

  \begin{thm}
    Strong detectability is decidable for LPNs.
  \end{thm}
  \begin{proof}
    The formula of Theorem~\ref{thm_formula} can be expressed as the following path formula:
    \begin{multline*}
    (\exists \sigma_1, \sigma_2, \sigma_3, \sigma_4)(\exists M_1,M_2,M_3,M_4) \\
      (M_{0,\|}
      \xrightarrow{\ \sigma_1\ }_{N_{\|}}  M_1
      \xrightarrow{\ \sigma_2\ }_{N_{\|}}  M_2
      \xrightarrow{\ \sigma_3\ }_{N_{\|}}  M_3
      \xrightarrow{\ \sigma_4\ }_{N_{\|}}  M_4)\\
    \land
      (M_2\le M_3)
      \land
      |\sigma_1|=0
      \land
      |\sigma_3|>0
      \land
      \bigvee_{p\in P} M_4(p)\neq M_4(p'),
    \end{multline*}
    where $|\sigma_1|=0$ is equivalent to $\land_{t\in T} \#_{t}(\sigma_1) \le 0$ and $|\sigma_3|>0$ is equivalent to $\lor_{t\in T} \#_{t}(\sigma_3) > 0$.
    Note that $M_4$ can be written as term $M_4-M_1+M_{0,\|}$,
    where $M_4-M_1$ and $M_{0,\|}$ are terms ($M_4$ and $M_1$ are marking variables but $M_{0,\|}$ is a constant).
    Therefore,
    the last term $\bigvee_{p\in P} M_4(p)\neq M_4(p')$ is equivalent to
     \[\bigvee_{p\in P}
     \left(\begin{array}{c c}
          &(M_4-M_1+M_{0,\|})(p)\!>\! (M_4-M_1+M_{0,\|})(p')\\
     \vee\!\!\! &(M_4-M_1+M_{0,\|})(p)\!<\!(M_4-M_1+M_{0,\|})(p')
     \end{array}
     \right),
     \] which is a valid predicate of Yen's path logic.
  \end{proof}

  Although the satisfiability of path formulae of Yen's logic is decidable, its complexity is open. There is a so-called {\em increasing\/} fragment of Yen's logic that requires that the path formula uses only marking predicates and $\varphi(M_1,\ldots,M_n,\sigma_1,\ldots,\sigma_n)$ implies that $M_n \ge M_1$. Deciding satisfiability of this fragment is EXPSPACE-complete~\cite{AtigH11}. However, the reader can see that our formula is not an increasing path formula, and hence the existing results do not imply any upper bound complexity.

  To discuss the lower bound complexity, we show that checking strong detectability requires at least exponential space. Our approach is to reduce the \emph{coverability problem}, which is know to be EXPSPACE-complete~\cite{esparza}.

  \begin{thm}
    Checking strong detectability is EXPSPACE-hard.
  \end{thm}
  \begin{proof}
    Given a Petri net system $(N,M_0)$, the coverability problem asks whether there is a reachable marking that covers a given marking $M$.

    Let $(N,M_0)$ and $M$ be the instance of the coverability problem. We construct a new Petri net as follows (see Fig.~\ref{pnhardproof} for an illustration).
    \begin{figure}
      \centering
      \includegraphics[scale=.75]{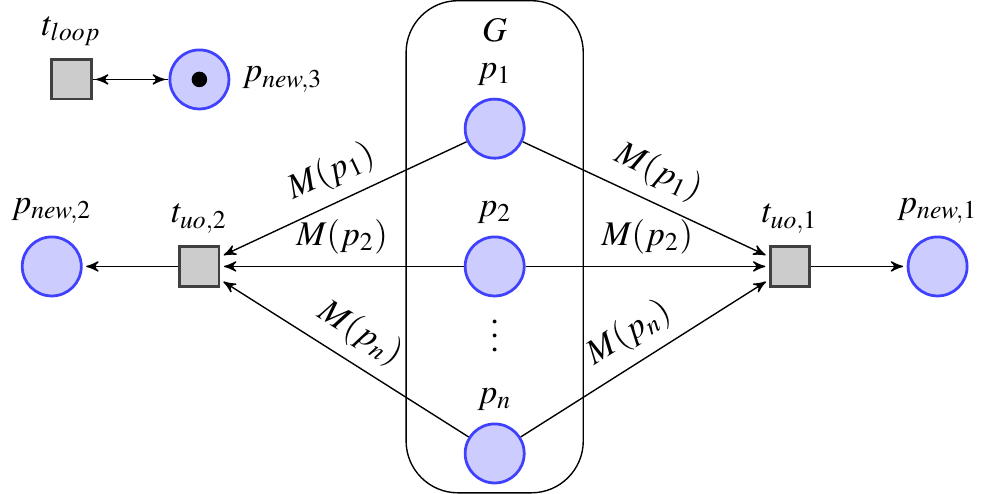}
      \caption{Sketch of the hardness construction}
      \label{pnhardproof}
    \end{figure}
    We add two new unobservable transitions $t_{uo,1}$ and $t_{uo,2}$, and two new place $p_{new,1}$ and $p_{new,2}$ initialized with zero tokens to $(N,M_0)$, and we define $Pre(p,t_{uo,1}) = Pre(p,t_{uo,2}) = M(p)$ for $p \in P$, and $Post(p_{new,i}, t_{uo,i})=1$ for $i=1,2$; unspecified mappings are defined as zero.
    We add a new isolated place $p_{new,3}$ initialized with one token, and define a new self-loop transition $t_{loop}$ in $p_{new,3}$ to guarantee that the system is deadlock free.
    Finally, we define the labeling function $\ell\colon T\cup\{t_{uo,1},t_{uo,2},t_{loop}\}\to T\cup\{t_{loop}\}$ by $\ell(t) = t$ for $t\in  T\cup\{t_{loop}\}$, and $\ell(t_{uo,1})=\ell(t_{uo,2}) = \eps$.

    By the construction, unobservable transitions $t_{uo,1}$ and $t_{uo,2}$ can be fired if and only if  $M$ can be covered.
    Thus, if these two unobservable transitions are firable, then the modified system is not strongly detectable because we cannot distinguish between the tokens in $p_{new,1}$ and $p_{new,2}$.
    On the other hand, if these two unobservable transitions are not firable, then all firable transitions are observable, which directly implies that the system is strongly detectable.
    Overall, the original system  covers $M$ if and only if the modified system is strongly detectable.   Hence, deciding strong detectability is EXPSPACE-hard.
  \end{proof}

\section{Weak Detectability}\label{sec:4}

In some applications, we only need to determine, after a finite number of observations, the current and subsequent states for some trajectories of the system.
This property is referred to as weak detectability and is defined as follows.

\begin{defn}
    An LPN system $G=(N,M_0,\Sigma,\ell)$ is {\em weakly detectable\/} if there exists an integer $n \ge 0$ and a word $s \in L^{\omega}(G)$ such that $|R(G,s')|=1$ for any prefix $s'$ of $s$ of length at least $n$.
\end{defn}

Deciding weak detectability for DES modeled by finite-state automata is a PSPACE-complete problem.
We now show that it is undecidable for DES modeled by unbounded LPNs.

\begin{thm}\label{thm5}
  Weak detectability is undecidable for LPNs.
\end{thm}
\begin{proof}
Let $G_1$ and $G_2$ be two  LPNs with no unobservable transitions, i.e., $\ell(t)$ is not the empty word for any transition $t$.
It is well-known that the inclusion problem, which asks whether $L(G_1)\subseteq L(G_2)$, is undecidable~\cite{Hack76} for LPNs even when all transitions are observable.
Next, we reduce the inclusion problem to the weak detectability verification problem.

  From $G_1$ and $G_2$, we construct an LPN $G$ as follows.
  \begin{figure}
    \centering
    \includegraphics[scale=.8]{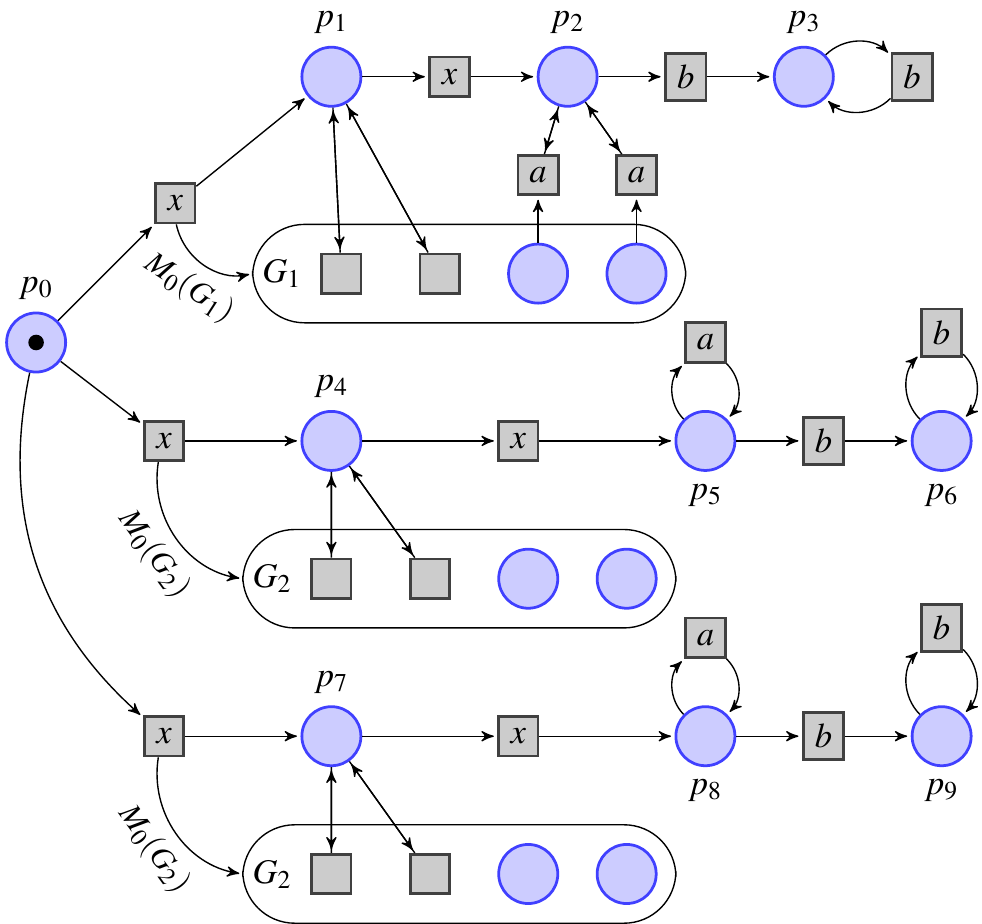}
    \caption{Sketch of the construction; labels depicted in transitions}
    \label{fig2}
  \end{figure}
  We create 10 new places $p_0$ up to $p_9$, and we use new labels $x$, $a$, and $b$ as depicted in Fig.~\ref{fig2}. Place $p_1$ (resp. $p_4$, $p_7$) is connected by a self-loop to every transition of $G_1$ (resp. $G_2$). Intuitively, $p_1$ (resp. $p_4$, $p_7$) allows $G$ to simulate $G_1$ (resp. $G_2)$. For every place of $G_1$, we create a new transition labeled by $a$ to which the place is connected, and through which there is a self-loop from place $p_2$ back to place $p_2$. The intuition is that $p_2$ allows $G$ to remove tokens from the $G_1$ part under a word from $a^*$. The rest of the Petri net $G$ is as depicted in Fig.~\ref{fig2}.

  The initial marking of $G$ consists of a single token in place $p_0$. At the beginning, only the transitions connected to place $p_0$ are enabled. Then, after the first transition (which is labeled by $x$), the net $G$ simulates either $G_1$ or $G_2$ from their corresponding initial markings, and hence the $\omega$-language of $G$ is
  \begin{align*}
    \{xwxa^{y(w)}b^{\omega} \mid w\in L(G_1)\}
    & \cup \{xwx(a^{\omega}+a^*b^{\omega}) \mid w\in L(G_2)\} \\
    & \cup \{xw \mid w\in L^{\omega}(G_1)\cup L^{\omega}(G_2)\}
  \end{align*}
  where
  $y(w)$ is finite and depends on the number of tokens in the net $G_1$ after generating the word $w\in L(G_1)$.

  We show that $L(G_1)\subseteq L(G_2)$ if and only if $G$ is not weakly detectable.

  If $L(G_1)\nsubseteq L(G_2)$, then there exists a word $w \in L(G_1)-L(G_2)$. We now consider all markings of $G_1$ after generating the word $w$. There can be several, but a finite number of such markings, because the length of $w$ is finite and there are no transitions labeled by $\eps$ in $G_1$. We sum the tokens in every such marking and let $k$ denote its maximum. This means that after generating $xwxa^kb$, the marking of $G$ is such that a single token is in place $p_3$, no tokens are in the part of $G_1$, because $k$ is the maximum number of tokens in $G_1$ after generating $w$, so we had to use all of them to generate $a^k$, and the part of $G_2$ contains no tokens. If the net now keeps generating $b^{\omega}$, we stay in this marking for ever. This is the only marking reachable by the $\omega$-word $xwxa^kb^{\omega}$, because $w\notin L(G_2)$. Thus, the net is weakly detectable; the $n$ from the definition is $n = |xwx|+k+1$, which is a constant for such a fixed word $w$.

  If $L(G_1)\subseteq L(G_2)$, then any word $xvxa^ub^{\omega}$ generated using the part with $G_1$, that is, $v\in L(G_1)$ and $u$ is bounded by the number of tokens in any marking of $G_1$ reachable after generating $v$ in $G_1$, can be simulated using the part of $G_2$. Moreover, any word from $\{xwx(a^{\omega}+a^*b^{\omega}) \mid w\in L(G_2)\} \cup \{xw \mid w\in L^{\omega}(G_2)\}$ generated by the part using $G_2$ always leads to at least two different markings because of the two identical parts in $G$ simulating $G_2$, cf. the places $p_4,p_5,p_6$ and $p_7,p_8,p_9$, and hence $G$ is not weakly detectable.
\end{proof}

\subsection{Application to Opacity}
  Opacity is a property related to the privacy and security analysis. The system has a secret modeled as a set of markings and an intruder is modeled as a passive observer with limited observation. The system is opaque if the intruder never knows for sure that the system is in a secret marking. We first recall the definition of opacity for LPNs~\cite{BryansKR05,TongLSG17}.

  \begin{defn}
    Let $G = (N, M_0, \Sigma, \ell)$ be an LPN system and $S \subseteq R(N,M_0)$. System $G$ is {\em current-state opaque\/} with respect to $S$ if for every $M \in S$ and $\sigma \in L(N,M_0)$ such that $M_0 \xrightarrow{\sigma} M$, there exists $\sigma' \in L(N,M_0)$ such that $\ell(\sigma') = \ell(\sigma)$ and $M_0 \xrightarrow{\sigma'} M'$ with $M' \notin S$.
  \end{defn}

  Informally, an LPN system is current-state opaque if for every transition sequence $\sigma$ leading to a marking in the secret set, there is another transition sequence $\sigma'$ whose firing leads to a non-secret marking, and the sequences produce the same observation $\ell(\sigma) = \ell(\sigma')$.

  Tong et al.~\cite{TongLSG17} showed that deciding current-state opacity of an LPN system is undecidable. In their proof, they reduce the inclusion problem for LPNs (is $L(G_1)\subseteq L(G_2)$?) and construct a secret set as large as the reachability set of $G_1$, which is infinite in general. It is a natural question whether undecidability of current-state opacity follows from the infinity of the secret set. Equivalently stated, the question is whether current-state opacity becomes decidable if the secret set is finite. As a consequence of our result, we show that it is not the case, since current-state opacity is undecidable even if the secret set consists of a single marking. This result strengthens and completes the study of Tong et al.~\cite{TongLSG17}.

  \begin{defn}
    Let $G = (N, M_0, \Sigma, \ell)$ be an LPN system, and let $M_s \in R(N,M_0)$ be a secret marking. System $G$ is {\em single-marking current-state opaque\/} with respect to $M_s$ if it is current-state opaque with respect to the set $\{M_s\}$.
  \end{defn}

The following is   a consequence of the proof of Theorem~\ref{thm5}.
  \begin{cor}
    Single-marking current-state opacity for LPNs is undecidable.
  \end{cor}
  \begin{proof}
    Consider the net $G$ constructed in the proof of Theorem~\ref{thm5}, and let the secret set consist of the marking $M_s$ having a single token in place $p_3$ and no tokens in other places. Then, $G$ is current-state opaque with respect to the secret set $\{M_s\}$ if and only if $G$ is not weakly detectable.
  \end{proof}

\section{Conclusions}\label{sec:5}\vspace{-10pt}
  We investigated the existence of algorithms to decide strong and weak detectability for LPNs. We showed that whereas there is an algorithm checking strong detectability, but this algorithm is infeasible because it requires at least exponential space, there is no algorithm checking weak detectability. We also discussed the question whether the undecidability of current-state opacity follows from the possibly infinite secret set and, as a consequence of our results, we showed that it is not the case. Namely, current-state opacity remains undecidable even if the secret set is a singleton.

  Besides strong and weak detectability, there are other notions of detectability proposed in the literature, such as initial-state detectability~\cite{shu2013detectability}, generalized detectability~\cite{ShuLin2011} or delayed detectability~\cite{shu2013delayed}. Investigating the verification of these variants for LPNs is an interesting future direction.

\appendix
\section{Complexity of the Assumption}\label{appA}
  Here we discuss the complexity of checking that the system does not generate an infinite unobservable sequence and show that it is EXPSPACE-complete.
  Given a net, the property can be expressed in Yen's path logic as a sequence $M_0\xrightarrow{s_1} M_1 \xrightarrow{s_2} M_2$ such that $M_1\le M_2 \land s_2\neq\lambda\land \ell(s_2)=\varepsilon$. Eliminating the transition predicate according to Yen's Lemma~3.2 of \cite{yen1992unified}  results in an increasing path formula~\cite{AtigH11}, and hence the satisfiability of this formula is in EXPSPACE.
  To show EXPSPACE-hardness, we reduce the coverability problem. Let $G$ be an LPN and $M$ be a marking. We modify $G$ by adding an unobservable transition that is a self-loop requiring all and exactly the tokens of $M$ to fire, returning the tokens back to $M$. Then $M$ is coverable in $G$ if and only if the modified net has an infinite sequence of unobservable transitions (the added unobservable self-loop).

\bibliographystyle{plain}
\bibliography{biblio}

\begin{thebibliography}{10}

\bibitem{AtigH11}
M.~F. Atig and P.~Habermehl.
\newblock On {Y}en's path logic for {P}etri nets.
\newblock {\em Int. J. Found. Comput. Sci.}, 22(4):783--799, 2011.

\bibitem{basile2015state}
F.~Basile, M.~P. Cabasino, and C.~Seatzu.
\newblock State estimation and fault diagnosis of labeled time {P}etri net
  systems with unobservable transitions.
\newblock {\em {IEEE} Trans. Autom. Control}, 60(4):997--1009, 2015.

\bibitem{BryansKR05}
J.~Bryans, M.~Koutny, and P.~Y.~A. Ryan.
\newblock Modelling opacity using {P}etri nets.
\newblock {\em Electron. Notes Theor. Comput. Sci.}, 121:101--115, 2005.

\bibitem{cabasino2012new}
M.~P. Cabasino, A.~Giua, S.~Lafortune, and C.~Seatzu.
\newblock A new approach for diagnosability analysis of {P}etri nets using
  verifier nets.
\newblock {\em {IEEE} Trans. Autom. Control}, 57(12):3104--3117, 2012.

\bibitem{cabasino2017marking}
M.~P. Cabasino, C.~N. Hadjicostis, and C.~Seatzu.
\newblock Marking observer in labeled {P}etri nets with application to
  supervisory control.
\newblock {\em {IEEE} Trans. Autom. Control}, 62(4):1813--1824, 2017.

\bibitem{dickson1913finiteness}
L.~Dickson.
\newblock Finiteness of the odd perfect and primitive abundant numbers with $n$
  distinct prime factors.
\newblock {\em Amer. J. Math.}, 35(4):413--422, 1913.

\bibitem{dotoli2009line}
M.~Dotoli, M.~P. Fanti, A.~M. Mangini, and W.~Ukovich.
\newblock On-line fault detection in discrete event systems by {P}etri nets and
  integer linear programming.
\newblock {\em Automatica}, 45(11):2665--2672, 2009.

\bibitem{esparza}
J.~Esparza.
\newblock Petri nets.
\newblock Lecture Notes, 2018.

\bibitem{GiuaSeatzu2002}
A.~Giua and C.~Seatzu.
\newblock Observability of place/transition nets.
\newblock {\em {IEEE} Trans. Autom. Control}, 47(9):1424--1437, 2002.

\bibitem{Hack76}
M.~Hack.
\newblock {\em Decidability questions for {P}etri {N}ets}.
\newblock PhD thesis, MIT, Cambridge, MA, {USA}, 1976.

\bibitem{keroglou2015detectability}
C.~Keroglou and C.~N. Hadjicostis.
\newblock Detectability in stochastic discrete event systems.
\newblock {\em Syst. Control Lett.}, 84:21--26, 2015.

\bibitem{keroglou2017verification}
C.~Keroglou and C.~N. Hadjicostis.
\newblock Verification of detectability in probabilistic finite automata.
\newblock {\em Automatica}, 86:192--198, 2017.

\bibitem{koenig}
D.~K\"onig.
\newblock {\"U}ber eine {S}chlussweise aus dem {E}ndlichen ins {U}nendliche.
\newblock {\em Acta Sci. Math.}, 3:121--130, 1927.

\bibitem{Masopust2017d}
T.~Masopust.
\newblock Complexity of deciding detectability in discrete event systems.
\newblock {\em Automatica}, 93:257--261, 2018.

\bibitem{OzverenW1990}
C.~M. Ozveren and A.~S. Willsky.
\newblock Observability of discrete event dynamic systems.
\newblock {\em {IEEE} Trans. Autom. Control}, 35(7):797--806, 1990.

\bibitem{Peterson1981}
J.~L. Peterson.
\newblock {\em {P}etri Net Theory and the Modeling of Systems}.
\newblock Prentice Hall, NJ, USA, 1981.

\bibitem{Ramadge1986}
P.~J. Ramadge.
\newblock Observability of discrete event systems.
\newblock In {\em Conference on Decision and Control (CDC)}, pages 1108--1112,
  1986.

\bibitem{ramirez2003observability}
A.~Ram{\'\i}rez-Trevi{\~n}, I.~Rivera-Rangel, and E.~L{\'o}pez-Mellado.
\newblock Observability of discrete event systems modeled by interpreted
  {P}etri nets.
\newblock {\em {IEEE} Trans. Robot. Autom.}, 19(4):557--565, 2003.

\bibitem{ru2010sensor}
Y.~Ru and C.~Hadjicostis.
\newblock Sensor selection for structural observability in discrete event
  systems modeled by {P}etri nets.
\newblock {\em {IEEE} Trans. Autom. Control}, 55(8):1751--1764, 2010.

\bibitem{ShuLin2011}
S.~Shu and F.~Lin.
\newblock Generalized detectability for discrete event systems.
\newblock {\em Syst. Control Lett.}, 60(5):310--317, 2011.

\bibitem{shu2013delayed}
S.~Shu and F.~Lin.
\newblock Delayed detectability of discrete event systems.
\newblock {\em {IEEE} Trans. Autom. Control}, 58(4):862--875, 2013.

\bibitem{shu2013detectability}
S.~Shu and F.~Lin.
\newblock I-detectability of discrete-event systems.
\newblock {\em {IEEE} Trans. Autom. Sci. Eng.}, 10(1):187--196, 2013.

\bibitem{ShuLinYing2007}
S.~Shu, F.~Lin, and H.~Ying.
\newblock Detectability of discrete event systems.
\newblock {\em {IEEE} Trans. Autom. Control}, 52(12):2356--2359, 2007.

\bibitem{tong2016equivalence}
Y.~Tong, Z.~Li, and A.~Giua.
\newblock On the equivalence of observation structures for {P}etri net
  generators.
\newblock {\em {IEEE} Trans. Autom. Control}, 61(9):2448--2462, 2016.

\bibitem{TongLSG17}
Y.~Tong, Z.~Li, C.~Seatzu, and A.~Giua.
\newblock Decidability of opacity verification problems in labeled {P}etri net
  systems.
\newblock {\em Automatica}, 80:48--53, 2017.

\bibitem{yen1992unified}
H.-C. Yen.
\newblock A unified approach for deciding the existence of certain {{P}etri}
  net paths.
\newblock {\em Inform. and Comput.}, 96(1):119--137, 1992.

\bibitem{yin2017initial}
X.~Yin.
\newblock Initial-state detectability of stochastic discrete-event systems with
  probabilistic sensor failures.
\newblock {\em Automatica}, 80:127--134, 2017.

\bibitem{yin2018prognosis}
X.~Yin.
\newblock Verification of prognosability for labeled {P}etri nets.
\newblock {\em {IEEE} Trans. Autom. Control}, 63(6):1828--1834, 2018.

\bibitem{yin2017decidability}
X.~Yin and S.~Lafortune.
\newblock On the decidability and complexity of diagnosability for labeled
  {P}etri nets.
\newblock {\em {IEEE} Trans. Autom. Control}, 62(11):5931--5938, 2017.

\bibitem{YinLafortune17}
X.~Yin and S.~Lafortune.
\newblock Verification complexity of a class of observational properties for
  modular discrete events systems.
\newblock {\em Automatica}, 83:199--205, 2017.

\bibitem{Zhang17}
K.~Zhang.
\newblock The problem of determining the weak (periodic) detectability of
  discrete event systems is {PSPACE}-complete.
\newblock {\em Automatica}, 81:217--220, 2017.

\bibitem{ZhangGiua2018}
K.~Zhang and A.~Giua.
\newblock Weak (approximate) detectability of labeled {P}etri net systems with
  inhibitor arcs.
\newblock In {\em WODES}, pages 179--183, 2018.

\end{thebibliography}

\end{document}